\documentclass[a4paper,11pt]{article}

\usepackage{enumitem}
\usepackage{float}
\usepackage[margin=2.5cm]{geometry}
\usepackage[grumpy]{gitinfo}
\usepackage{graphicx}
\usepackage{times}
\usepackage[hyphens]{url}
\usepackage{xspace}

\usepackage{amsmath}
\allowdisplaybreaks[2]

\usepackage{amssymb}

\usepackage{amsthm}
\newtheoremstyle{plain-boldhead}
  {\topsep}
  {\topsep}
  {\itshape}
  {}
  {\bfseries}
  {.}
  { }
  {\thmname{#1}\thmnumber{ #2}\thmnote{ (\bfseries #3)}}
\newtheoremstyle{definition-boldhead}
  {\topsep}
  {\topsep}
  {\normalfont}
  {}
  {\bfseries}
  {.}
  { }
  {\thmname{#1}\thmnumber{ #2}\thmnote{ (\bfseries #3)}}
\theoremstyle{plain-boldhead}
\newtheorem{theorem}{Theorem}

\theoremstyle{definition-boldhead}
\newtheorem{definition}{Definition}

\floatstyle{ruled}
\newfloat{algo}{htbp}{algo}
\floatname{algo}{Algorithm}

\newcommand{\str}[1]{\textsc{#1}}
\newcommand{\var}[1]{\textit{#1}}
\newcommand{\op}[1]{\textsl{#1}}
\newcommand{\msg}[2]{\ensuremath{[\str{#1}, {#2}]}}
\newcommand{\becomes}{\ensuremath{\leftarrow}}

\newcommand{\nil}{\ensuremath{\bot}}
\newcommand{\event}[1]{\mbox{\textsl{#1}}}
\newcommand{\eventp}[2]{\mbox{\textsl{#1}(#2)}}
\newcommand{\false}{\str{False}\xspace}
\newcommand{\true}{\str{True}\xspace}

\newcommand{\CE}{\ensuremath{\mathcal{E}}\xspace}
\newcommand{\CI}{\ensuremath{\mathcal{I}}\xspace}
\newcommand{\CM}{\ensuremath{\mathcal{M}}\xspace}
\newcommand{\CO}{\ensuremath{\mathcal{O}}\xspace}
\newcommand{\CP}{\ensuremath{\mathcal{P}}\xspace}
\newcommand{\CS}{\ensuremath{\mathcal{S}}\xspace}
\newcommand{\CU}{\ensuremath{\mathcal{U}}\xspace}

\newcommand{\shortlist}{\parskip0pt\topsep0pt\partopsep0pt\parsep1pt\itemsep2pt}

\providecommand{\nondet}{non-de\-ter\-min\-istic\xspace}

\begin{document}

\title{\bf Non-determinism in Byzantine Fault-Tolerant Replication%
\footnote{A preliminary version of this work has been presented
at OPODIS 2016 (Madrid, Spain) and an extended abstract appears in 
the proceedings.}}

\author{Christian Cachin \and Simon Schubert \and Marko Vukoli\'c}

\date{IBM Research - Zurich\\
  \url{(cca|sis|mvu)@zurich.ibm.com}\\[2ex]
  \gitCommitterDate}

\maketitle

\begin{abstract}\noindent
  Service replication distributes an application over many processes for
  tolerating faults, attacks, and misbehavior among a subset of the
  processes.  With the recent interest in blockchain technologies,
  distributed execution of one logical application has become a prominent
  topic.  The established state-machine replication paradigm inherently
  requires the application to be deterministic.  This paper distinguishes
  three models for dealing with non-determinism in replicated services,
  where some processes are subject to faults and arbitrary behavior
  (so-called Byzantine faults): first, the modular case that does not
  require any changes to the potentially \nondet application (and neither
  access to its internal data); second, master-slave solutions, where ties
  are broken by a leader and the other processes validate the choices of
  the leader; and finally, applications that use cryptography and secret
  keys.  Cryptographic operations and secrets must be treated specially
  because they require strong randomness to satisfy their goals.

  The paper also introduces two new protocols.  First,
  Protocol~\emph{Sieve} uses the modular approach and filters out
  \nondet operations in an application.  It ensures that
  all correct processes produce the same outputs and that their internal
  states do not diverge.  A second protocol, called \emph{Mastercrypt},
  implements cryptographically secure randomness generation with a
  verifiable random function and is appropriate for most situations in
  which cryptographic secrets are involved.  All protocols are described in
  a generic way and do not assume a particular implementation of the
  underlying consensus primitive.
\end{abstract}

\pagestyle{plain}

\section{Introduction}

State-machine replication is an established way to enhance the resilience
of a client-server application~\cite{schnei90}.  It works by executing the
service on multiple independent components that will not exhibit correlated
failures.  We consider the approach of \emph{Byzantine fault-tolerance
  (BFT)}, where a group of \emph{processes} connected only by an unreliable
network executes an application~\cite{peshla80}.  The processes use a
protocol for \emph{consensus} or \emph{atomic broadcast} to agree on a
common sequence of operations to execute.  If all processes start from the
same initial state, if all operations that modify the state are
\emph{deterministic}, and if all processes execute the same sequence of
operations, then the states of the correct processes will remain the same.
(This is also called \emph{active} replication~\cite{cbpesc10}.)
A client executes an operation on the service by sending the operation to
all processes; it obtains the correct outcome based on comparing the
responses that it receives, for example, by a relative majority among the
answers or from a sufficiently large set of equal responses.  Tolerating
\emph{Byzantine faults} means that the clients obtain correct outputs as
long as a qualified majority of the processes is correct, even if the
faulty processes behave in arbitrary and adversarial ways.

Traditionally state-machine replication requires the application to be
deterministic.  But many applications contain implicit or explicit
non-determinism: in multi-threaded applications, the scheduler may
influence the execution, input/output operations might yield different
results across the processes, probabilistic algorithms may access a
random-number generator, and some cryptographic operations are inherently
not deterministic.  

Recently BFT replication has gained prominence because it may implement
distributed consensus for building
\emph{blockchains}~\cite{bmcnkf15,cdegjk16,Vukolic16,dccl16}.  A blockchain
provides a distributed, append-only ledger with cryptographic verifiability
and is governed by decentralized control.  It can be used to record events,
trades, or transactions immutably and permanently and forms the basis for
cryptocurrencies, such as Bitcoin or Ripple, or for running ``smart
contracts,'' as in Ethereum.
With the focus on active replication, this work aims at \emph{permissioned}
blockchains, which run among known entities~\cite{swanso15}.  In contrast,
\emph{permissionless} blockchains (including Ethereum) do not rely on
identities and use other approaches for reaching consensus, such as
proof-of-work protocols.
For practical use of blockchains, ensuring
deterministic operations is crucial since even the smallest divergence
among the outputs of different participants lets the blockchain diverge (or
``fork'').

This work presents a \emph{general treatment} of non-determinism in the
context of BFT replication and introduces a distinction among different
models to tackle the problem of non-determinism.  For example,
applications involving cryptography and secret encryption keys should be
treated differently from those that access randomness for other goals.  We
also distinguish whether the replication mechanism has access to the
application's source code and may modify it.

We also introduce two novel protocols.  The first, called \emph{Sieve},
replicates \nondet programs using in a \emph{modular} way,
where we treat the application as a black box and cannot change it.  We
target workloads that are usually deterministic, but which may occasionally
yield diverging outputs.  The protocol initially executes all operations
speculatively and then compares the outputs across the processes.  If the
protocol detects a minor divergence among a small number of processes, then
we \emph{sieve out the diverging values}; if a divergence among too many
processes occurs, we \emph{sieve out the operation} from sequence.
Furthermore, the protocol can use \emph{any} underlying consensus primitive
to agree on an ordering.
The second new protocol, \emph{Mastercrypt}, provides master-slave
replication with cryptographic security from verifiable random functions.
It addresses situations that require strong, cryptographically secure
randomness, but where the faulty processes may leak their secrets.

\subsection{Contributions}

We introduce three different models and discuss corresponding protocols for
replicating \nondet applications.

\begin{description}\shortlist
\item[Modular:] When the application itself is fixed and cannot be changed,
  then we need \emph{modular} replicated execution.  In practice this is
  often the case.  We distinguish two approaches for integrating a
  consensus protocol for ordering operations with the replicated execution
  of operations.  One can either use \emph{order-then-execute}, where the
  operations are ordered first, executed independently, and the results are
  communicated to the other processes through atomic broadcast.  This
  involves only deterministic steps and can be viewed as ``agreement on the
  input.''  Alternatively, with \emph{execute-then-order}, the processes
  execute all operations speculatively first and then ``agree on the
  output'' (of the operation).  In this case operations with diverging
  results may have to be rolled back.

  We introduce Protocol~\emph{Sieve} that uses speculative execution and
  follows the \emph{execute-then-order} approach.  As described before,
  \emph{Sieve} is intended for applications with occasional
  non-determinism.  It represents the first modular solution to replicating
  \nondet applications in a BFT system.

\item[Master-slave:] In the \emph{master-slave} model, one process is
  designated as the master or ``leader,'' makes all \nondet
  choices that come up, and imposes these on the others which act as slaves
  or ``followers.''  Because a faulty (Byzantine) master may misbehave, the
  slaves must be able to validate the selections of the master before the
  operation can be executed as determined by the master.  The master-slave
  model is related to passive replication; it
  works for most applications including probabilistic algorithms, but
  cannot be applied directly for cryptographic operations.  As a further
  complication, this model requires that the developer has access to the
  internals of the application and can modify it.

  For the master-slave model we give a detailed description of the
  well-known replication protocol, which has been used in earlier systems.

\item[Cryptographically secure:] Traditionally, randomized applications can
  be made deterministic by deriving pseudorandom bits from a secret seed,
  which is initially chosen truly randomly.  Outsiders, such as clients of
  the application, cannot distinguish this from an application that uses
  true randomness.  This approach does not work for BFT replication, where
  faulty processes might expose and leak the seed.  To solve this problem,
  we introduce a novel protocol for master-slave replication with
  cryptographic randomness, abbreviated~\emph{Mastercrypt}.  It
  lets the master select random bits with a \emph{verifiable random
    function}.  The protocol is aimed at applications that need strong,
  cryptographically secure randomness; however it does not protect against
  a faulty master that leaks the secret.  We also review the established
  approach of threshold (public-key) cryptography, where private keys are
  secret-shared among the processes and cryptographic operations are
  distributed in a fault-tolerant way over the whole group.
\end{description}
The modular Protocol~\emph{Sieve} has been developed for running
potentially \nondet smart contracts as applications on top of a
permissioned blockchain platform, built using BFT replication.  An
implementation has been made available as open source in ``Hyperledger
fabric'' \footnote{\url{https://github.com/hyperledger/fabric}}, which is part of
the Linux Foundation's Hyperledger Project.  As of November 2016, the
project has decided to adopt a different architecture
\footnote{\url{https://github.com/hyperledger/fabric/blob/master/proposals/r1/Next-Consensus-Architecture-Proposal.md}}; the
platform has been redesigned to use a master-slave approach for addressing
\nondet execution.

\subsection{Related work}

The problem of ensuring deterministic operations for replicated services is
well-known.  When considering only crash faults, many authors have
investigated methods for making services deterministic, especially for
multi-threaded, high-performance services~\cite{bresch96}.  Practical
systems routinely solve this problem today using master-slave replication,
where the master removes the ambiguity and sends deterministic updates to
the slaves.  In recent research on this topic, for instance, Kapitza et
al.~\cite{kscsd10} present an optimistic solution for making multithreaded
applications deterministic.  Their solution requires a predictor for
\nondet choices and may invoke additional communication via the
consensus module.

In the BFT model, most works consider only sequential execution of
deterministic commands, including PBFT~\cite{caslis02} and
UpRight~\cite{cklwad09}.  BASE~\cite{caroli03} and CBASE~\cite{kotdah04}
address Byzantine faults and adopt the master-slave model for handling
non-determinism, focusing on being generic (BASE) and on achieving high
throughput (CBASE), respectively.  These systems involve changes to the
application code and sometimes also need preprocessing steps for
operations.

Fault-tolerant execution on multi-core servers poses a new challenge, even
for deterministic applications, because thread-level parallelism may
introduce unpredictable differences between processes.  Eve~\cite{kwqcad12}
heuristically identifies groups of non-interfering operations and executes
each group in parallel.  Afterwards it compares the outputs, may roll back
operations that lead to diverging states, or could transfer an agreed-on
result state to diverging processes.  Eve resembles Protocol~\emph{Sieve}
in this sense, but lacks modularity.

For the same domain of scalable services running on multi-cores,
Rex~\cite{ghyzzz14} uses the master-slave model, where the master executes
the operations first and records its \nondet choices.  The slaves replay
these operations and use a consensus primitive to agree on a consistent
outcome.  Rex only tolerates crashes, but does not address the BFT model.

Fault-tolerant replication involving cryptographic secrets and distributed
cryptography has been pioneered by Reiter and Birman~\cite{reibir94}.  Many
other works followed, especially protocols using threshold cryptography; an
early overview of solutions in this space was given by
Cachin~\cite{cachin01}.

In current work Duan and Zhang~\cite{haizha16} discuss how the master-slave
approach can handle randomized operations in BFT replication, where
execution is separated from agreement in order to protect the privacy of
the data and computation.

\subsection{Organization}
\label{subsec:org}

The remainder of this paper starts with Section~\ref{sec:def}, containing
background information and formal definitions of broadcast, replication,
and atomic broadcast (i.e., consensus).  The following sections contain the
discussion and protocols for the three models: the modular solution
(Section~\ref{sec:modular}), the master-slave protocol
(Section~\ref{sec:master}), and replication methods for applications
demanding cryptographic security (Section~\ref{sec:secure}).

\section{Definitions}
\label{sec:def}

\subsection{System model}

We consider a distributed system of \emph{processes} that
communicate with each other and provide a common \emph{service} in a
fault-tolerant way.  Using the paradigm of service
replication~\cite{schnei90}, requests to the service are broadcast among
the processes, such that the processes execute all requests in the same
order.  The clients accessing the service are not modeled here.  We denote
the set of processes by \CP and let $n = |\CP|$.  A process may be
\emph{faulty}, by crashing or by exhibiting \emph{Byzantine faults}; the
latter means they may deviate arbitrarily from their specification.
Non-faulty processes are called \emph{correct}.  Up to $f$ processes may be
faulty and we assume that~$n > 3f$.  The setup is also called a
\emph{Byzantine fault-tolerant (BFT) service replication system} or simply
a \emph{BFT system}.

We present protocols in a modular way using an event-based
notation~\cite{CachinGR11}.  A process is specified through its
\emph{interface}, containing the events that it exposes to other processes,
and through a set of \emph{properties}, which define its behavior.  A
process may react to a received event by doing computation and triggering
further events.
The events of a process interface consist of \emph{input events}, which the
process receives from other processes, typically to invoke its services,
and \emph{output events}, through which the process delivers information or
signals a condition to another process.

Every two processes can \emph{send} messages to each other using an
authenticated point-to-point communication primitive.  When a message
arrives, the receiver learns also which process has sent the message.  The
primitive guarantees \emph{message integrity}, i.e., when a message~$m$ is
received by a correct process with indicated sender~$p_s$, and $p_s$ is
correct, then $p_s$ previously sent~$m$.  Authenticated communication can
be implemented easily from an insecure communication channel by using a
mes\-sage-authen\-ti\-ca\-tion code (MAC)~\cite{MenezesOV97}, a symmetric
cryptographic primitive that relies on a secret key shared by every pair of
processes.  These keys have been distributed by a trusted entity
beforehand.

The system is \emph{partially synchronous}~\cite{dwlyst88} in the sense
that there is no a priori bound on message delays and the processes have no
synchronized clocks, as in an asynchronous system.  However, there is a
time (not known to the processes) after which the system is \emph{stable}
in the sense that message delays and processing times are bounded.  In
other words, the system is \emph{eventually synchronous}.  This model
represents a broadly accepted network model and covers a wide range of
real-world situations.

\subsection{Broadcast and state-machine replication}
\label{subsec:broadcast}

Suppose $n$ processes participate in a broadcast primitive.  Every process
may \emph{broadcast} a request or message~$m$ to the others.  The
implementation generates events to output the requests when they have been
agreed; we say that a request~$m$ is \emph{delivered} through this.  Atomic
broadcast also solves the \emph{consensus}
problem~\cite{hadtou93,CachinGR11}.  We use a variant that delivers only
messages satisfying a given \emph{external validity}
condition~\cite{ckps01}.

\begin{definition}[Byzantine atomic broadcast with external validity]
  \label{def:abv}
  A \emph{Byzantine atomic broadcast with external validity} (\emph{abv})
  is defined with the help of a validation predicate~$V()$ and in terms of
  these events:
  \begin{description}\shortlist
  \item[Input event:] \eventp{abv-broadcast}{$m$}: Broadcasts a
    message $m$ to all processes.
  \item[Output event:] \eventp{abv-deliver}{$p, m$}: Delivers a
    message $m$ broadcast by process $p$.
  \end{description}

  The deterministic predicate~$V()$ validates messages.  It can be computed
  locally by every process.  It ensures that a correct process only
  delivers messages that satisfy~$V()$.  More precisely, $V()$ must
  guarantee that when two correct processes~$p$ and $q$ have both delivered
  the same sequence of messages up to some point, then $p$ obtains $V(m) =
  \true$ for any message~$m$ if and only if $q$ also determines that~$V(m)
  = \true$.

  With this validity mechanism, the broadcast satisfies:
  \begin{description}\shortlist
  \item[Validity:] If a correct process~$p$ broadcasts a message~$m$, then
    $p$ eventually delivers~$m$.
  \item[External validity:] When a correct process delivers some
    message~$m$, then $V(m) = \true$.
  \item[No duplication:] No correct process delivers the same message more
    than once.
  \item[Integrity:] If some correct process delivers a message~$m$ with
    sender~$p$ and process~$p$ is correct, then $m$ was previously
    broadcast by~$p$.
  \item[Agreement:] If a message~$m$ is delivered by some correct process,
    then $m$ is eventually delivered by every correct process.
  \item[Total order:] Let $m_1$ and $m_2$ be any two messages and suppose
    $p$ and $q$ are any two correct processes that deliver $m_1$ and $m_2$.
    If $p$ delivers $m_1$ before $m_2$, then $q$ delivers $m_1$ before
    $m_2$.
  \end{description}
\end{definition}

In practice it may occur that not all processes agree in the above sense on
the validity of a message.  For instance, some correct process
may conclude $V(m) = \true$ while others find that~$V(m) = \false$.  For this
case it is useful to reason with the following relaxation:
\begin{description}\shortlist
\item[Weak external validity:] When a correct process delivers some
  message~$m$, then at least one correct process has determined that $V(m)
  = \true$ at some time between when $m$ was broadcast and when it was
  delivered.
\end{description}
Every protocol for Byzantine atomic broadcast with external validity of
which we are aware either ensures this weaker notion or can easily be
changed to satisfy it.

Atomic broadcast is the main tool to implement state-machine replication
(SMR), which executes a service on multiple processes for tolerating
process faults.  Throughout this work we assume that many operation
requests are generated concurrently by all processes; in other words, there
is request contention.

A \emph{state machine} consists of variables and operations that transform
its state and may produce some output. Traditionally, operations are
\emph{deterministic} and the outputs of the state machine are solely
determined by the initial state and by the sequence of operations that it
has executed.

The state machine \emph{functionality} is defined by~$\op{execute}()$, a
function that takes a \emph{state} $s \in \CS$, initially~$s_0$, and
operation~$o \in \CO$ as input, and outputs a successor state~$s'$ and a
\emph{response} or \emph{output value}~$r$:
\[
  \op{execute}(s, o) \ \to \ (s', r),
\]

A \emph{replicated state machine} can be characterized as in
Definition~\ref{def:rsm}.  Basically, its interface presents two events:
first, an input event \eventp{rsm-execute}{\var{operation}} that a process
uses to invoke the execution of an operation~$o$ of the state machine; and
second, an output event \eventp{rsm-output}{$o, s, r$}, which is produced
by the state machine.  The output indicates the operation has been executed
and carries the resulting state~$s$ and response~$r$.  We assume here that
an operation~$o$ includes both the name of the operation to be executed and
any relevant parameters.

\begin{definition}\label{def:rsm}
  A \emph{replicated state machine (rsm)} for a functionality
  $\op{execute}()$ and initial state~$s_0$ is defined by these events:
  \begin{description}\shortlist
  \item[Input event:] \eventp{rsm-execute}{$o$}: Requests that the state
    machine executes the operation~$o$.
  \item[Output event:] \eventp{rsm-output}{$o, s, r$}: Indicates that the
    state machine has executed an operation~$o$, resulting in new
    state~$s$, and producing response~$r$.
  \end{description}
  \noindent
  It also satisfies these properties:
  \begin{description}\shortlist
  \item[Agreement:] The sequences of executed operations and corresponding
    outputs are the same for all correct processes.
  \item[Correctness:] When a correct process has executed a sequence of
    operations~$o_1, \dots, o_k$, then the sequences of output states~$s_1,
    \dots, s_k$ and responses~$r_1, \dots, r_k$ satisfies for $i=1,\dots,
    k$,
    \[
      (s_i, r_i) \ = \ \op{execute}(s_{i-1}, o_i)
    \]
  \item[Termination:] If a correct process executes a operation, then the
    operation eventually generates an output.
  \end{description}
\end{definition}

The standard implementation of a replicated state machine relies on an
atomic broadcast protocol to disseminate the requests to all
processes~\cite{schnei90,hadtou93}.  Every process starts from the same
initial state and executes all operations in the order in which they are
delivered.  If all operations are \emph{deterministic} the states of the
correct processes never diverge.

\subsection{Leader election}
\label{subsec:leaderelection}

Implementations of atomic broadcast need to make some
synchrony assumptions or employ randomization~\cite{filypa85}.  A very weak
timing assumption that is also available in many practical implementations
is an \emph{eventual leader-detector oracle}~\cite{chatou96,hadtou93}.

We define an eventual leader-detector primitive, denoted~$\Omega$, for a
system with Byzantine processes.  It informs the processes about one
correct process that can serve as a leader, so that the protocol can
progress.  When faults are limited to crashes, such a leader detector can be
implemented from a failure detector~\cite{chatou96}, a primitive that, in
practice, exploits timeouts and low-level point-to-point messages to
determine whether a remote process is alive or has crashed.

With processes acting in arbitrary ways, though, one cannot rely on the
timeliness of simple responses for detecting Byzantine faults.  One needs
another way to determine remotely whether a process is faulty or performs
correctly as a leader.  Detecting misbehavior in this model depends
inherently on the specific protocol being executed~\cite{dggs99}.  We use
the approach of ``trust, but verify,'' where the processes monitor the
leader for correct behavior.  More precisely, a leader is chosen
arbitrarily, but ensuring a fair distribution among all processes (in fact,
it is only needed that a correct process is chosen at least with constant
probability on average, over all leader changes).  Once elected, the chosen
leader process gets a chance to perform well.  The other processes monitor
its actions.  Should the leader not have achieved the desired goal after
some time, they complain against it, and initiate a switch to a new leader.

Hence we assume that the leader should act according to the application and
within some time bounds. If the leader performs wrongly or exceeds the
allocated time before reaching this goal, then other processes detect this
and report it as a failure to the leader detector by filing a complaint.
In an asynchronous system with eventual synchrony as considered here, every
process always behaves according to the specification and eventually all
remote processes also observe this; if such correct behavior cannot be
observed from a process, then the process must be faulty.

This notion of ``performance'' depends on the specific algorithm
executed by the processes, which relies on the output from the
leader-detection module.  Therefore, eventual leader election with
Byzantine processes is not an isolated low-level abstraction, as with
crash-stop processes, but requires some input from the higher-level
algorithm.  The \eventp{$\Omega$-complain}{$p$} event allows to express
this.  Every process may \emph{complain} against the current leader~$p$ by
triggering this event.

\begin{definition}[Byzantine leader detector]
  \label{def:bld}
  A \emph{Byzantine leader detector} ($\Omega$) is defined with these
  events:
  \begin{description}\shortlist
  \item[Output event:] \eventp{$\Omega$-trust}{$p$}: Indicates that
    process~$p$ is trusted to be leader.
  \item[Input event:] \eventp{$\Omega$-complain}{$p$}: Expresses a
    complaint about the performance of leader process~$p$.
  \end{description}

  The primitive satisfies the following properties:
  \begin{description}\shortlist
  \item[Eventual accuracy:] There is a time after which every correct
    process trusts some correct process.
  \item[Eventual succession:] If more than $f$ correct processes that trust
    some process~$p$ complain about~$p$, then every correct process
    eventually trusts a different process than~$p$.
  \item[Coup resistance:] A correct process~$q$ does not trust a new leader
    unless at least one correct process has complained against the leader
    which~$q$ trusted before.
  \item[Eventual agreement:] There is a time after which no two correct
    processes trust different processes.
  \end{description}
\end{definition}

It is possible to lift the output from the Byzantine leader detector to an
\emph{epoch-change} primitive, which outputs
not only the identity of a leader but also an increasing \emph{epoch
  number}.  This abstraction divides time into a series of
epochs at every participating process, where epochs are identified by
numbers.  The numbers of the epochs started by one particular process
increase monotonically (but they do not have to form a complete sequence).
Moreover, the primitive also assigns a \emph{leader} to every epoch, such
that any two correct processes in the same epoch receive the same leader.
The mechanism for processes to complain about the leader is the same as
for~$\Omega$.

More precisely, epoch change is defined as follows~\cite{CachinGR11}:

\begin{definition}[Byzantine epoch-change]
  \label{def:bec}
  A \emph{Byzantine epoch-change} ($\Psi$) primitive is defined with these
  events:
  \begin{description}\shortlist
  \item[Output event:] \eventp{$\Psi$-start-epoch}{$e, p$}: Indicates that
    the epoch with number~$e$ and leader~$p$ starts.
  \item[Input event:] \eventp{$\Psi$-complain}{$e, p$}: Expresses a
    complaint about the performance of leader process~$p$.
  \end{description}
  
  The primitive satisfies the following properties:
  \begin{description}\shortlist
  \item[Monotonicity:] If a correct process starts an epoch $(e, p)$
    and later starts an epoch $(e', p')$, then $e' > e$.
  \item[Consistency:] If a correct process starts an epoch $(e, p)$ and
    another correct process starts an epoch $(e', p')$ with $e = e'$,
    then $p = p'$.
  \item[Eventual succession:] Suppose more than $f$ correct processes have
    started an epoch~$(e, p)$ as their last epoch; when these processes all
    complain about~$p$, then every correct process eventually starts an
    epoch with a number higher than~$e$.
  \item[Coup resistance:] When a correct process that has most recently
    started some epoch~$(e_1, p_1)$ starts a new epoch~$(e_2, p_2)$, then
    at least one correct process has complained about leader~$p_1$ in
    epoch~$e_1$.
  \item[Eventual leadership:] There is a time after which every correct
    process has started some epoch and starts no further epoch, such that
    the last epoch started at every correct process is epoch~$(e, p)$ and
    process~$p$ is correct.
  \end{description}
\end{definition}

When an epoch-change abstraction is initialized, it is assumed that a
default epoch with number~0 and a leader~$p_0$ has been started at all
correct processes.  The value of $p_0$ is made available to all processes
implicitly.  All ``practical'' BFT systems in the eventual-synchrony model
starting from PBFT~\cite{caslis02} implicitly contain an implementation of
Byzantine epoch-change; this notion was described explicitly by Cachin et
al.~\cite[Chap.~5]{CachinGR11}.

\subsection{Hash functions and digital signatures}

We model cryptographic \emph{hash functions} and \emph{digital signature
  schemes} as ideal, deterministic functionalities implemented by a
distributed oracle.

A cryptographic \emph{hash function} maps a bit string of arbitrary length
to a short, unique representation.  The functionality provides only a
single operation \op{hash}; its invocation takes a bit string $x$ as
parameter and returns an integer $h$ with the response.  The implementation
maintains a list $L$ of all $x$ that have been queried so far.  When the
invocation contains $x\in L$, then \op{hash} responds with the index of $x$
in $L$; otherwise, \op{hash} appends $x$ to $L$ and returns its index.
This ideal implementation models only collision resistance but no other
properties of real hash functions.

The functionality of the \emph{digital signature scheme} provides two
operations, $\op{sign}_p$ and $\op{verify}_p$.  The invocation of
$\op{sign}_p$ specifies a process~$p$, takes a bit string~$m$ as input, and
returns a signature $\sigma \in \{0,1\}^*$ with the response.  Only $p$ may
invoke $\op{sign}_p$.  The operation $\op{verify}_p$ takes a putative
signature~$\sigma$ and a bit string~$m$ as parameters and returns a Boolean
value with the response.  Its implementation satisfies that
$\op{verify}_p(\sigma ,m)$ returns \true for any process~$p$ and $m \in
\{0,1\}^*$ if and only if $p$ has executed $\op{sign}_p(m)$ and obtained
$\sigma$ before; otherwise, $\op{verify}_p(\sigma ,m)$ returns \false.
Every process may invoke \op{verify}.  The signature scheme may be
implemented analogously to the hash function.

\section{Modular protocol }
\label{sec:modular}

In this section we discuss the \emph{modular} execution of replicated
\nondet programs.  Here the program is given as a black box, it cannot be
changed, and the BFT system cannot access its internal data structures.
Very informally speaking, if some processes arrive at a different output
during execution than ``most'' others, then the output of the disagreeing
processes is discarded.  Instead they should ``adopt'' the output of the
others, e.g., by asking them for the agreed-on state and response.  When
the outputs of ``too many'' processes disagree, the correct output may not
be clear; the operation is then ignored (or, as an optimization,
quarantined as \nondet) and the state rolled back.  In this modular
solution any application can be replicated without change; the application
developers may not even be aware of potential non-determinism.  On the
other hand, the modular protocol requires that most operations are
deterministic and produce almost always the same outputs at all processes;
it would not work for replicating probabilistic functions.

More precisely, a \emph{\nondet state machine} may output
different states and responses for the same operation, which are due to
probabilistic choices or other non-repeatable effects.  Hence we assume
that \op{execute} is a relation and not a deterministic function, that is,
repeated invocations of the same operation with the same input may yield
different outputs and responses.  This means that the standard approach of
state-machine replication based directly on atomic broadcast fails.

There are two ways for modular black-box replication of
\nondet applications in a BFT system:
\begin{description}
\item[Order-then-execute:] Applying the SMR principle directly, the
  operations are first ordered by atomic broadcast.  Whenever a process
  delivers an operation according to the total order, it executes the
  operation.  It does not output the response, however, before checking
  with enough others that they all arrive at the same outputs.  To this
  end, every process atomically broadcasts its outputs (or a hash of the
  outputs) and waits for receiving a given number (up to $n-f$) of outputs
  from distinct processes.  Then the process applies a fixed decision
  function to the atomically delivered outputs, and it determines the
  successor state and the response.

  This approach ensures consistency due to its conceptual simplicity but is
  not very efficient in typical situations, where atomic broadcast forms
  the bottleneck.  In particular, in atomic broadcast with external
  validity, a process can only participate in the ordering of the next
  operation when it has determined the outputs of the previous one.  This
  eliminates potential gains from pipelining and increases the overall
  latency.

\item[Execute-then-order:] Here the steps are inverted and the operations
  are executed \emph{speculatively} before the system commits their order.
  As in other practical protocols, this solution uses the heuristic
  assumption that there is a designated \emph{leader} which is usually
  correct.  Thus, every process sends its operations to the leader and the
  leader orders them.  It asks all processes to execute the operations
  speculatively in this order, the processes send (a hash of) their outputs
  to the leader, and the leader determines a unique output.  Note that this
  value is still speculative because the leader might fail or there might
  be multiple leaders acting concurrently.  The leader then tries to obtain
  a confirmation of its speculative order by atomically broadcasting the
  chosen output.  Once every process obtains this output from atomic
  broadcast, it commits the speculative state and outputs the response.

  In rare cases when a leader is replaced, some processes may have
  speculated wrongly and executed other operations than those determined
  through atomic broadcast.  Due to non-determinism in the execution a
  process may also have obtained a different speculative state and response
  than what the leader has obtained and broadcast.  This implies that the
  leader must either send the state (or state delta) and the response resulting from the
  operation though atomic broadcast, or that a process has a different way
  to recover the decided state from other processes.
\end{description}

In the following we describe Protocol~\emph{Sieve}, which adopts the
approach of \emph{execute-then-order} with speculative execution.

\subsection{Protocol \emph{Sieve}}
\label{subsec:sieve}

Protocol~\emph{Sieve} runs a Byzantine atomic broadcast with weak external
validity (abv) and uses a \emph{sieve-leader} to coordinate the execution
of \nondet operations.  The leader is elected through a Byzantine
epoch-change abstraction, as defined in
Section~\ref{subsec:leaderelection}, which outputs epoch/leader tuples with
monotonically increasing epoch numbers.  For the \emph{Sieve} protocol
these epochs are called \emph{configurations}, and \emph{Sieve} progresses
through a series of them, each with its own sieve-leader.

The processes send all operations to the service through the leader of the
current configuration, using an \str{invoke} message.  The current leader
then initiates that all processes execute the operation speculatively;
subsequently the processes agree on an output from the operation and
thereby \emph{commit} the operation.  As described here, \emph{Sieve}
executes one operation at a time, although it is possible to greatly
increase the throughput using the standard method of \emph{batching}
multiple operations together.

The leader sends an \str{execute} message to all processes with the
operation~$o$.  In turn, every process executes~$o$ \emph{speculatively} on
its current state~$s$, obtains the speculative next state~$t$ and the
speculative response~$r$, signs those values, and sends a hash and the
signature back to the leader in an \str{approve} message.

The leader receives $2f+1$ \str{approve} messages from distinct processes.
If the leader
observes at least $f+1$ approvals for the \emph{same} speculative output,
then it \emph{confirms} the operation and proceeds to committing and
executing it.  Otherwise, the leader concludes that the operation is
\emph{aborted} because of diverging outputs.  There must be $f+1$ equal
outputs for confirming~$o$, in order to ensure that every process will
eventually learn the correct output, see below.

The leader then \op{abv-broadcasts} an \str{order} message, containing the
operation, the speculative output $(t,r)$ for a confirmed operation or an
indication that it aborted, and for validation the set of \str{approve}
messages that justify the decision whether to confirm or abort.  During
atomic broadcast, the external validity check by the processes will verify
this justification.

\begin{algo*}
\vbox{
\small
\begin{tabbing}
xxxx\=xxxx\=xxxx\=xxxx\=xxxx\=xxxx\=MMMMMMMMMMMMMMMMMMM\=\kill
\textbf{State} \\

\> $\CI$: set of invoked operations at every process 
\` $B[p]$, for $p \in \CP$: buffer at sieve-leader \\
\> $\var{config}$: sieve-config number 
\` $\var{leader}$: sieve-leader, initially $p_0$ \\
\> $\var{next-epoch}$: announced sieve-config number, initially $\bot$ 
\` $\var{next-leader}$: announced sieve-leader, initially $\bot$ \\
\> $s$: current state, initially $s_0$ 
\` $\var{cur}$: current operation, initially \nil \\
\> $t$: speculative state, initially \nil 
\` $r$: speculative response, initially \nil \\
\\
\textbf{upon invocation} $\op{rsm-execute}(o)$ \textbf{do} \\
\> $\CI \becomes \CI \cup \{ o \}$ \\
\> send message \msg{invoke}{\var{config}, o} over 
   point-to-point link to \var{leader} \\
\\
\textbf{upon} receiving message \msg{invoke}{c, o} from $p$
   \textbf{such that} $B[p] = \nil$ 
   \textbf{and} $c = \var{config}$ 
   \textbf{and} $\var{leader} = \var{self}$ \textbf{do} \\
\> $B[p] \becomes o$
   \` // buffer only the latest operation from each process \\
\\
\textbf{upon} exists $p$ that $B[p] \neq \nil$
   \textbf{such that} $\var{cur} = \nil$
   \textbf{and} $\var{leader} = \var{self}$ \textbf{do} \\
\> $\var{cur} \becomes B[p]$ \\
\> send \msg{execute}{\var{config}, \var{cur}} over point-to-point links 
   to all processes \\
\\
\textbf{upon} receiving message \msg{execute}{c, o} from process~$p$
   \textbf{such that} $p = \var{leader}$ 
   \textbf{and} $c = \var{config}$
   \textbf{and} $t = \nil$ \textbf{do} \\
\> $(t, r) \becomes \op{execute}(s, o)$ \\
\> $\sigma \becomes \op{sign}_{\var{self}}(
      \str{speculate} \| \var{config} \| \op{hash}(t \| r) )$ \\
\> send message \msg{approve}{\var{config}, o, \op{hash}(t \| r), 
     \sigma} to \var{leader} \\
\\
\textbf{upon} receiving $2f+1$ messages
   \msg{approve}{c_p, o_p, h_p, \sigma_p}, each from a distinct process~$p$,
   \textbf{such that} \\
   \> \> $c_p = \var{config}$ 
   \textbf{and} $op_p = \var{cur}$
   \textbf{and} $\op{verify}_p(\sigma_p, \str{speculate} \| \var{config} \| h_p )$ 
   \textbf{and} $\var{leader} = \var{self}$ \textbf{do} \\
\> \textbf{if} there is a set \CE of $f+1$ received \str{approve} messages 
   whose $h_p$ value is equal to $\op{hash}(t\|r)$ \textbf{then} \\
\> \> \eventp{abv-broadcast}{\msg{order}{\str{confirm}, \var{config}, 
      \var{cur}, t, r, \CE}} \\
\> \textbf{else} \\
\> \> let \CU be the set of $2f+1$ received \str{approve} messages \\
\> \> \eventp{abv-broadcast}{\msg{order}{\str{abort}, \var{config}, 
      \var{cur}, \nil, \nil, \CU}} \\
\\
\textbf{upon} event \eventp{abv-deliver}{$p$, \msg{order}{\var{decision}, 
   c, o, t_c, r_c, \cdot}}
   \textbf{such that} $c = \var{config}$ \textbf{do} 
   \` // commit~$o$ \\
\> \textbf{if} $\var{leader} = \var{self}$ \textbf{then} \\
\> \> $B[p] \becomes \nil$ \\
\> \> $\var{cur} \becomes \nil$ \\
\> \textbf{if} $o \in \CI$ \textbf{then} \\
\> \> $\CI \becomes \CI \setminus \{o\}$ \\
\> \textbf{if} $\var{decision} = \str{confirm}$ \textbf{then} \\
\> \> $s \becomes t_c$ 
   \` // adopt the agreed-on state and response, needed 
      if $(t_c, r_c) \neq (t,r)$ \\
\> \> \eventp{rsm-output}{$o, s, r_c$}\\
\> $t \becomes \nil$ \\
\\
\textbf{upon} event \eventp{$\Psi$-start-epoch}{$e, p$} \textbf{do} \\
\> $(\var{next-epoch}, \var{next-leader}) \becomes (e, p)$ \\
\> \textbf{if} $p = \var{self} \land e > \var{config}$ \textbf{then} \\
\> \> \eventp{abv-broadcast}{\msg{new-sieve-config}{e, \var{self}}} \\
\\
\textbf{upon} event \eventp{abv-deliver}{$p$, \msg{new-sieve-config}{c, p}}
   \textbf{do} \\
\> $(\var{config}, \var{leader}) \becomes (c, p)$ \\
\> $t \becomes \nil$ \\
\\
\textbf{periodically do} \\
\> for every operation~$o \in \CI$, determine the age of~$o$ since it
  has been invoked and added to~\CI \\
\> \textbf{if} there are ``old'' operations in \CI \textbf{then} \\
\> \> \eventp{$\Psi$-\op{complain}}{\var{leader}}
\end{tabbing}
}
\caption{Protocol \emph{Sieve}: replicated state machine with
  non-deterministic operations}
\label{alg:sieve1}
\end{algo*}

\begin{algo*}
\vbox{
\small
\begin{tabbing}
xxxx\=xxxx\=xxxx\=xxxx\=xxxx\=xxxx\=\kill
\textbf{upon invocation} $V(m)$ 
  \textbf{do} \\
\> \textbf{if} $m = [\str{order}, \str{decision}, c, o, \CM]$ \textbf{then} \\
\> \> \textbf{if} \CM is a set of $f+1$ messages of the form
      \msg{approve}{c_p, o_p, h_p, \sigma_p} \textbf{such that} \\
\> \> \> $c_p = \var{config}$ \textbf{and} $o_p = o$ \textbf{and} 
         $\op{verify}_p(\sigma_p, \str{speculate}\|c_p\|h_p) = \true$ 
         \textbf{and} \\
\> \> \> all $h_p$ values in \CM are equal \textbf{then} \\
\> \> \textbf{return} \true \\
\> \textbf{else if} $m = [\str{order}, \str{abort}, c, o, \CM]$ \textbf{then} \\
\> \> \textbf{if} \CM is a set of $2f+1$ messages of the form
      \msg{approve}{c_p, o_p, h_p, \sigma_p} \textbf{such that} \\
\> \> \> $c_p = \var{config}$ \textbf{and} $o_p = o$ \textbf{and} 
         $\op{verify}_p(\sigma_p, \str{speculate}\|c_p\|h_p) = \true$ 
         \textbf{and} \\
\> \> \> no $f+1$ of the $h_p$ values in \CM are equal \textbf{then} \\
\> \> \textbf{return} \true \\
\> \textbf{else if} $m = [\str{new-sieve-config}, c, p]$ \textbf{then} \\
\> \> \textbf{if} $c \leq \var{next-epoch}$ \textbf{and} 
      $p = \var{next-leader}$ \textbf{then} \\
\> \> \> \textbf{return} \true \\
\> \textbf{return} \false
\end{tabbing}
}
\caption{Validation predicate $V()$ for Byzantine atomic broadcast
  used inside Algorithm~\emph{Sieve}}
\label{alg:sieve2}
\end{algo*}

As soon as an \str{order} message with operation~$o$ is
\emph{abv-delivered} to a process in \emph{Sieve}, $o$ is committed.  If
$o$ is confirmed, the process adopts the output decided by the leader.
Note this may differ from the speculative output computed by the process.
Protocol~\emph{Sieve} therefore includes the next state~$t$ and the
response~$r$ in the \str{order} message.  In practice, however, one might
not send~$t$, but state deltas, or even only the hash value of $t$ while
relying on a different way to recover the confirmed state.  Indeed, since
$f+1$ processes have approved any confirmed output, a process with a wrong
speculative output is sure to reach at least one of them for obtaining the
confirmed output later.

In case the leader \emph{abv-broadcasted} an \str{order} message with the
decision to abort the current operation because of the diverging
outputs (i.e., no $f+1$ identical hashes in $2f+1$ \str{approve} messages),
the process simply ignores the current request and speculative
state. As an optimization, processes may \emph{quarantine} the current
request and flag it as \nondet.

As described so far, the protocol is open to a denial-of-service attack by
multiple faulty processes disguising as sieve-leaders and executing
different operations.  Note that the epoch-change abstraction, in periods
of asynchrony, will not ensure that any two correct processes agree on the
leader, as some processes might skip configurations.  Therefore
\emph{Sieve} also orders the configuration and leader changes using
consensus (with the \emph{abv} primitive).

To this effect, whenever a process receives a \op{start-epoch} event with
itself as leader, the process \op{abv-broadcasts} a \str{new-sieve-config}
message, announcing itself as the leader.  The validation predicate for
broadcast verifies that the leader announcement concerns a configuration
that is not newer than the most recently started epoch at the validating
process, and that the process itself endorses the same next leader.  Every
process then starts the new configuration when the \str{new-sieve-config}
message is \op{abv-delivered}.  If there was a speculatively executed
operation, it is aborted and its output discarded.

The design of \emph{Sieve} prevents uncoordinated speculative request
execution, which may cause contention among requests from different
self-proclaimed leaders and can prevent liveness easily.  Naturally, a
faulty leader may also violate liveness, but this is not different from
other leader-based BFT protocols.

The details of Protocol~\emph{Sieve} are shown in
Algorithms~\ref{alg:sieve1}--\ref{alg:sieve2}.  The pseudocode assumes that
all point-to-point messages among correct processes are authenticated,
cannot be forged or altered, and respect FIFO order.  The invoked
operations are unique across all processes and \var{self} denotes the
identifier of the executing process.

\subsection{Correctness}

\begin{theorem}
  \label{thm:sieve}
  Protocol~\emph{Sieve} implements a replicated state machine allowing a
  \nondet functionality $\op{execute}()$, except that demonstrably \nondet
  operations may be filtered out and not executed.
\end{theorem}

\begin{proof}[Proof sketch]
  The \emph{agreement} condition of Definition~\ref{def:rsm} follows
  directly from the protocol and from the \emph{abv} primitive.  Every
  \event{rsm-output} event is immediately preceded by an \op{abv-delivered}
  \str{order} message, which is the same for all correct processes due to
  \emph{agreement} of \emph{abv}.  Since all correct processes react to it
  deterministically, their outputs are the same.

  For the \emph{correctness} property, note that the outputs $(s_i, r_i)$
  (state and response) resulting from an operation~$o$ must have been
  confirmed by the protocol and therefore the values were included in an
  \str{approve} message from at least one correct process.  This process
  computed the values such that they satisfy $(s_i, r_i) =
  \op{execute}(s_{i-1}, o)$ according to the protocol for handling an
  \str{execute}~message.  On the other hand, no correct process outputs
  anything for committed operations that were aborted, this is permitted by
  the exception in the theorem statement.  Moreover, only operations are
  filtered out for which distinct correct processes computed diverging
  outputs, as ensured by the sieve-leader when it determines whether the
  operation is confirmed or aborted.  In order to abort, no set of $f+1$
  processes must have computed the same outputs among the $2f+1$ processes
  sending the \str{approve}~messages.  Hence, at least two among every set
  of~$f+1$ correct processes arrived at diverging outputs.

  \emph{Termination} is only required for deterministic operations, they
  must terminate despite faulty processes that approve wrong outputs.  The
  protocol ensures this through the condition that at least~$f+1$ among the
  $2f+1$ \str{approve} messages received by the sieve-leader are equal.
  The faulty processes, of which there are at most~$f$, cannot cause an
  abort through this.  But every \str{order} message is eventually
  \op{abv-delivered} and every confirmed operation is eventually executed
  and generates an output.
\end{proof}

\subsection{Optimizations}
\label{subsec:sieveopt}

\paragraph{Rollback and state transfer.}  In Protocol~\emph{Sieve} every
process maintains a copy of the application state resulting from an
operation until the operation is committed.  Moreover, the confirmed state
and the response of an operation are included in \str{order} messages that
are \op{abv-broadcast}.  For practical applications though, this is often
too expensive, and another way for recovering the application state is
needed.  The solution is to roll back an operation and to transfer the
correct state from other processes.

We assume that there exists a \emph{rollback} primitive, ensuring that for
$(t, r) \gets \op{execute}(s, o)$, the output of $\op{rollback}(t, r, o)$
is always~$s$.  In the \str{order} messages of the protocol, the resulting
output state and response are replaced by their hashes for checking the
consistency among the processes.  Thus, when a process receives an
\str{order} message with a confirmed operation and hashes~$ht$ and $hr$ of
the output state and response, respectively, it checks whether the
speculative state and response satisfy $\op{hash}(t) = ht$ and
$\op{hash}(r) = hr$.  If so, it proceeds as in the protocol.

If the committed operation was aborted, or the values do not match in a
confirmed operation, the process rolls back the operation.  For a confirmed
operation, the process then invokes \emph{state transfer} and retrieves the
correct state~$t$ and response~$r$ that match the hashes from other
clients.  It will then set the state variable~$s$ to $t$ and output the
response~$r$.  Rollback helps to implement transfer state efficiently, by
sending incremental updates only.

For transferring the state, the process sends a \str{state-request} message
to all those processes who produced the \str{speculate}-signatures
contained in the $f+1$ \str{approve} messages, which the process receives
together with a committed and confirmed operation.  Since at most $f$ of
them may fail to respond, the process is guaranteed to receive the correct
state.

State transfer is also initiated when a new configuration starts through an
\op{abv-delivered} \str{new-sieve-config} message, but the process has
already speculatively executed an operation in the last configuration
without committing it (this can be recognized by $t \neq \nil$).  As in the
above use of state transfer, the operation must terminate before the
process becomes ready to execute further operations from \str{execute}
messages.

\paragraph{Synchronization with PBFT-based atomic broadcast.}

When the well-known \emph{PBFT protocol}~\cite{caslis02} is used to
implement \op{abv-broadcast}, two further optimizations are possible,
mainly because PBFT also relies on a leader and already includes Byzantine
epoch-change.  Hence assume that every process runs PBFT according to
Castro and Liskov~\cite[Sec.~4]{caslis02}.

First, let an epoch-change event of $\Psi$ occur at every view change of
PBFT.  More precisely, whenever a process has received a correct
PBFT-\textsc{new-view} message for PBFT-view number~$v$ and new primary
process~$p$, and when the matching \textsc{view-change} messages have
arrived, then the process triggers a \eventp{$\Psi$-start-epoch}{$v, p$}
event at \emph{Sieve}.  The process subsequently starts executing requests
in the new view.  Moreover, complaints from \emph{Sieve} are handled in the
same way as when a backup in PBFT \emph{suspects} the current primary to be
faulty, namely, it initiates a view change by sending a \str{view-change}
message.

The view change mechanism of PBFT ensures all properties expected
from~$\Psi$, as follows.  The \emph{monotonicity} and \emph{consistency}
properties of Byzantine epoch-change follow directly from the calculation
of strictly increasing view numbers in PBFT and the deterministic
derivation of the PBFT-primary from the view.  The \emph{eventual
  leadership} condition follows from the underlying timing assumption,
which essentially means that timeouts are increased until, eventually,
every correct process is able to communicate with the leader, the leader is
correct, and no further epoch-changes occur.

The second optimization concerns the \str{new-sieve-config} message.
According to \emph{Sieve} it is \op{abv-broadcast} whenever a new
sieve-leader is elected, by that leader.  As the leader is directly mapped
to PBFT's primary here, it is now the primary who sends this message as a
PBFT request.  Note that this request might not be delivered when the
primary fails, but it will be delivered by the other processes according to
the properties of \op{abv-broadcast}, as required by \emph{Sieve}.  Hence,
the new sieve-configuration and sieve-leader are assigned either by all
correct processes or by none of them.

With these two specializations for PBFT, \emph{Sieve} incurs the additional
cost of the \str{execute}/\str{approve} messages in the request flow, and
one \str{new-sieve-config} following every view-change of PBFT.  But
determining the sieve-leader and implementing $\Psi$ do not lead to any
additional messages.

\subsection{Discussion}

Non-deterministic operations have not often been discussed in the context
of BFT systems.  The literature commonly assumes that deterministic
behavior can be imposed on an application or postulates to change the
application code for isolating non-determinism.  In practice, however, it
is often not possible.

Liskov~\cite{liskov10} sketches an approach to deal with non-determinism in
PBFT which is similar to \emph{Sieve} in the sense that it treats the
application code modularly and uses execute-then-order.  This proposal is
restricted to the particular structure of PBFT, however, and does not
consider the notion of external validity for \emph{abv} broadcast.

For applications on multi-core servers, the \emph{Eve}
system~\cite{kwqcad12} also executes operation groups speculatively across
processes and detects diverging states during a subsequent verification
stage.  In case of divergence, the processes must roll back the operations.
The approach taken in Eve resembles that of \emph{Sieve}, but there
are notable differences. Specifically, the primary application of Eve continues to assume deterministic
operations, and non-determinism may only result from concurrency during
parallel execution of requests.  Furthermore, this work uses a particular
agreement protocol based on PBFT and not a generic \emph{abv} broadcast
primitive.

It should be noted that \emph{Sieve} not only works with Byzantine atomic
broadcast in the model of eventual synchrony, but can equally well be run
over randomized Byzantine consensus~\cite{ckps01,mxcss16}.

\section{Master-slave protocol}
\label{sec:master}

By adopting the \emph{master-slave} model one can support a broader range
of \nondet application behavior compared to the modular protocol.
This design generally requires source-code access and modifications to the
program implementing the functionality.  In a master-slave protocol for
\nondet execution, one process is designated as \emph{master}.
The master executes every operation first and records all \nondet
choices.  All other processes act as \emph{slaves} and follow the same
choices.  To cope with a potentially Byzantine master, the slaves must be
given means to verify that the choices made by the master are plausible.
The master-slave solution presented here follows \emph{primary-backup
  replication}~\cite{bmst93}, which is well-known to handle
\nondet operations.  For instance, if the application accesses a
pseudorandom number generator, only the master obtains the random bits from
the generator and the slaves adopt the bits chosen by the master.  This
protocol does not work for functionalities involving cryptography, however,
where master-slave replication typically falls short of achieving the
desired goals.  Instead a cryptographically secure protocol should be used;
they are the subject of Section~\ref{sec:secure}.

\subsection{Non-deterministic execution with evidence}
\label{subsec:nondet}

As introduced in Section~\ref{sec:modular}, the \op{execute} operation of a
\nondet state machine is a relation.  Different output values are
possible and represent acceptable outcomes.  We augment the output of an
operation execution by adding \emph{evidence} for justifying the resulting
state and response.  The slave processes may then \emph{replay} the choices
of the master and accept its output.

More formally, we now extend \op{execute} to \op{nondet-execute} as
follows:
\[
  \op{nondet-execute}(s, o) \ \to \ (s', r, \rho).
\]
Its parameters $s$, $o$, $s'$, and $r$ are the same as for \op{execute};
additionally, the function also outputs \emph{evidence}~$\rho$.  Evidence
enables the slave processes to execute the operation by themselves and
obtain the same output as the master, or perhaps only to validate the
output generated by another execution.  For this task there is a function
\[
  \op{verify-execution}(s, o, s', r, \rho) \ \to \ \{\false, \true\}
\]
that outputs \true if and only if the set of possible outputs from
$\op{nondet-execute}(s, o)$ contains $(s', r, \rho)$.  For completeness we
require that for every $s$ and $o$, when $(s', r, \rho) \gets
\op{nondet-execute}(s, o)$, it always holds $\op{verify-execute}(s, o, s',
r, e)$ $=$ $\true$.

As a basic verification method, a slave could rerun the computation of the
master.  Extensions to use cryptographic verifiable
computation~\cite{walblu15} are possible.  Note that we consider randomized
algorithms to be a special case of \nondet ones.  The evidence for
executing a randomized algorithm might simply consist of the random coin
flips made during the execution.

\subsection{Replication protocol}
\label{subsec:master-protocol}

Implementing a replicated state machine with \nondet operations
using master-slave replication does not require an extra round of messages
to be exchanged, as in Protocol~\emph{Sieve}.  It suffices that the master
is chosen by a Byzantine epoch-change abstraction and that the master
broadcasts every operation together with the corresponding evidence.

More precisely, the processes operate on top of an underlying broadcast
primitive \emph{abv} and a Byzantine epoch-change abstraction~$\Psi$.
Whenever a process receives a \op{start-epoch} event with itself as leader
from~$\Psi$, the process considers itself to be the master for the epoch
and \op{abv-broadcasts} a message that announces itself as the master for
the epoch.  The epochs evolve analogously to the configurations in
\emph{Sieve}, with the same mechanism to approve changes of the master in
the validation predicate of atomic broadcast.  Similarly, non-master
processes send their operations to the master of the current epoch for
ordering and execution.

For every invoked operation~$o$, the master computes $(s', r, \rho) \gets
\op{nondet-execute}(s, o)$ and \op{abv-broadcasts} an \str{order} message
containing the current epoch~$c$ and parameters $o$, $s'$, $r$, and~$\rho$.
The validation predicate of atomic broadcast for \str{order} messages
verifies that the message concerns the current epoch and that
$\op{verify-execution}(s, o, s', r, \rho) = \true$ using the current state~$s$
of the process.  Once an \str{order} message is \op{abv-delivered}, a
process adopts the response and output state from the message as its own.

As discussed in the first optimization for \emph{Sieve}
(Section~\ref{subsec:sieveopt}), the output state~$s'$ and response~$r$ do
not always have to be included in the \str{order} messages.  In the
master-slave model, they can be replaced by hashes only for those
operations where the evidence~$\rho$ contains sufficient data for a process
to compute the same $s'$ and $r$ values as the master.  This holds, for
example, when all \nondet choices of an operation are contained in~$\rho$.

Should the master \op{abv-broadcast} an operation with evidence that does
not execute properly, i.e., $\op{verify-execution}(s, o, s', r, \rho) =
\false$, the atomic broadcast primitive ensures that it is not
\op{abv-delivered} through the external validity property.  As in
\emph{Sieve}, every process periodically checks if the operations that it
has invoked have been executed and complains against the current master
using~$\Psi$.  This ensures that misbehaving masters are eventually
replaced.

\subsection{Discussion}

The master-slave protocol is inspired by primary-backup
replication~\cite{bmst93}, and for the concrete scenario of a BFT system,
it was first described by Castro, Rodrigues, and Liskov in
BASE~\cite{caroli03}.  The protocol of BASE addresses only the particular
context of PBFT, however, and not a generic atomic broadcast primitive.

As mentioned before, the master-slave protocol requires changes to the
application for extracting the evidence that will convince the slave
processes that choices made by the master are valid.  This works well in
practice for applications in which only a few, known steps can lead to
divergence.  For example, operations reading inputs from the local system,
accessing platform-specific environment data, or generating randomness can
be replicated whenever those functions are provided by programming
libraries.  Master-slave replication may only be employed when the
application developer is aware of the causes of non-determinism; for
example, a multi-threaded application influenced by a \nondet
scheduler could not be replicated unless the developer can also control the
scheduling (e.g.,~\cite{kscsd10}).

\section{Cryptographically secure protocols}
\label{sec:secure}

Security functions implemented with cryptography are more important today
than ever.  Replicating an application
that involves a cryptographic secret, however, requires a careful
consideration of the attack model.  If the BFT system should tolerate that
$f$ processes become faulty in arbitrary ways, it must be assumed that
their secrets leak to the adversary against whom the cryptographic scheme
is employed.  

Service-level secret keys must be protected and should never leak to
an individual process.  Two solutions have been explored to address this
issue.  One could delegate this responsibility to a third party, such as a
centralized service or a secure hardware module at every process.  However,
this contradicts the main motivation behind replication: to eliminate
central control points.  Alternatively one may use \emph{distributed
  cryptography}~\cite{desmed94}, share the keys among the processes so that
no coalition of up to $f$ among them learns anything, and perform the
cryptographic operations under distributed control.  This model was
pioneered by Reiter and Birman~\cite{reibir94} and exploited, for instance,
by SINTRA~\cite{cachin01,cacpor02} or COCA~\cite{zhscre02a}.

In this section we discuss two methods for integrating non-deterministic
cryptographic operations in a BFT system.  The first scheme is a novel
protocol in the context of BFT systems, called \emph{Mastercrypt},
and uses verifiable random functions to generate
pseudorandom bits.  This randomness is unpredictable and cannot be biased by a
Byzantine process.  The second scheme is the well-known technique of
distributed cryptography, as discussed above, which addresses a broad range
of cryptographic applications.  Both schemes adopt the master-slave
replication protocol from the previous section.

\subsection{Randomness from verifiable random functions}
\label{subsec:vrf}

A \emph{verifiable random function (VRF)}~\cite{mirava99} resembles a
pseudorandom function but additionally permits anyone to verify
non-interactively that the choice of random bits occurred correctly.  The
function therefore guarantees correctness for its output without disclosing
anything about the secret seed, in a way similar to non-interactive
zero-knowledge proofs of correctness.

More precisely, the process owning the VRF chooses a secret seed~$sk$ and
publishes a public verification key~$vk$.  Then the function family
$G_{sk}: \{0,1\}^\lambda \to \{0,1\}^\mu$ and algorithms $P_{sk}$ and
$V_{pk}$ are a VRF whenever three properties hold:
\begin{description}\shortlist
\item[Correctness:] $y \becomes G_{sk}(x)$ can be computed efficiently from
  $sk$ and for every~$x$ one can also (with the help of~$sk$) efficiently
  generate a proof $\pi \becomes P_{sk}(x)$ such that $V_{pk}(x,y,\pi) =
  \true$.
\item[Uniqueness:] For every input~$x$ there is a unique~$y$ that satisfies
  $V_{pk}(x,y,\pi)$, i.e., it is impossible to find $y_0$ and $y_1 \neq
  y_0$ and $\pi_0$ and $\pi_1$ such that $V_{pk}(x,y_0,\pi_0) =
  V_{pk}(x,y_1,\pi_1) = \true$.
\item[Pseudorandomness:] From knowing $vk$ alone and sampling values from
  $V_{sk}$ and $P_{sk}$, no polynomial-time adversary can distinguish the
  output of $G_{sk}(x)$ from a uniformly random $\mu$-bit string, unless
  the adversary calls the owner to evaluate~$V_{sk}$ or $P_{sk}$ on~$x$.
\end{description}
Thus, a VRF generates a value for every input~$x$ which is unpredictable
and pseudorandom to anyone \emph{not} knowing~$sk$.  As $G_{sk}(x)$ is
unique for a given~$x$, even an adversarially chosen key preserves the
pseudorandomness of $G$'s outputs towards other processes.

Efficient implementations of VRFs have not been easy to find, but the
literature nowadays contains a number of reasonable constructions under
broadly accepted hardness assumptions~\cite{lysyan02,jager15}.  In
practice, when adopting the random-oracle model, VRFs can immediately be
obtained from unique signatures such as ordinary RSA
signatures~\cite{lysyan02}.

\paragraph{Replication with cryptographic randomness from a VRF.}
With master-slave replication, cryptographically strong
randomness secure against faulty non-leader processes can be obtained from
a VRF as follows.  Initially every process generates a VRF-seed and a
verification key. Then it passes the verification key to a trusted entity,
which distributes the $n$ verification keys to all processes consistently,
ensuring that all correct processes use the same list of verification keys.
At every place where the application needs to generate (pseudo-)randomness,
the VRF is used by the master to produce the random bits and all processes
verify that the bits are unique.

In more detail, \emph{Mastercrypt} works as follows.
The master computes all random choices while executing an
operation~$o$ as $r \becomes G_{sk}(\var{tag})$, where \var{tag} denotes a
unique identifier for the instance and operation.  This tag must never
reused by the protocol and should not be under the control of the master.
The master supplies $\pi \becomes P_{sk}(\var{tag})$ to the other processes
as evidence for the choice of~$r$.  During the verification step in
$\op{verify-execution}()$ every process now validates that
$V_{pk}(\var{tag},r,\pi) = \true$.  When executing the operation, every
process uses the same randomness~$r$.

The pseudorandomness property of the VRF ensures that no process apart from
the master (or anyone knowing its secret seed) can distinguish $r$ from
truly random bits.  This depends crucially on the condition that \var{tag}
is used only once as input to the VRF.  Hence this solution yields a
deterministic pseudorandom output that achieves the desired
unpredictability and randomness in many cases, especially against entities
that are not part of the BFT system.  Note that simply handing over the
seed of a cryptographic pseudorandom generator to all processes and
treating the generator as part of a deterministic application would be
predictable for the slave processes and not pseudorandom.

Of course, if the master is faulty then it can predict the value of~$r$,
leak it to other processes, and influence the protocol accordingly.  The
protocol should leave as little as possible choice to the master for
influencing the value of~\var{tag}.  It could be derived from an identifier
of the protocol or BFT system ``instance,'' perhaps including the
identities of all processes, followed by a uniquely determined binary
representation of the operation's sequence number.  If one assumes all
operations are represented by unique bit strings, a hash of the operation
itself could also serve as identifier.

\subsection{Using distributed cryptography}
\label{subsec:distcrypto}

\emph{Distributed cryptography} or, more precisely, \emph{threshold
  cryptography}~\cite{desmed94} distributes the power of a cryptosystem
among a group of $n$ processes such that it tolerates $f$ faulty ones,
which may leak their secrets, fail to participate in protocols, or even act
adversarially against the other processes.  Threshold cryptosystems extend
cryptographic secret sharing schemes, which permit the process group to
maintain a secret such that $f$ or fewer of them have no information about
it, but any set of \emph{more} than $f$ can reconstruct it.

A \emph{threshold public-key cryptosystem (T-PKCS)}, for example, is a
public-key cryptosystem with distributed control over the decryption
operation.  There is a single public key for encryption, but each process
holds a \emph{key share} for decryption.  When a ciphertext is to be
decrypted, every process computes a decryption share from the ciphertext
and its key share.  From any $f+1$ of these decryption shares, the
plaintext can be recovered.  Usually the decryption shares are accompanied
by zero-knowledge proofs to make the scheme robust.  This models a
non-interactive T-PKCS, which is practical because it only needs one round
of point-to-point messages for exchanging the decryption shares; other
T-PKCSs require interaction among the processes for computing shares.  The
public key and the key shares are generated either by a trusted entity
before the protocol starts or again in a distributed way, tolerating faulty
processes that may try to disrupt the key-generation
protocol~\cite{gjkr07}.

A state-of-the-art T-PKCS is secure against \emph{adaptive
  chosen-ciphertext attacks}~\cite{shogen98}, ensuring that an adversary
cannot obtain any meaningful information from a ciphertext unless at least
one correct process has computed a decryption share.  With a T-PKCS the BFT
system can receive operations in encrypted form, order them first without
knowing their content, and only decrypt and execute them after they have
been ordered.  This approach defends against violations of the causal order
among operations~\cite{reibir94,ckps01}.

A \emph{threshold signature scheme} works analogously and can be used, for
instance, to implement a secure name service or a certification
authority~\cite{cachin01,zhscre02a,cacsam04}.  Practical non-interactive
threshold signature schemes are well-known~\cite{shoup00a}.  To generate
cryptographically strong and unpredictable pseudorandom bits,
\emph{threshold coin-tossing schemes} have also been
constructed~\cite{cakush05}.  They do not suffer from the limitation of the
VRF construction in the previous section and ensure that no single process
can predict the randomness until at least one correct process has agreed to
start the protocol.

\paragraph{Replication with threshold cryptosystems.}
Threshold cryptosystems have been used in BFT replication starting with the
work Reiter and Birman~\cite{reibir94}.  Subsequently
SINTRA~\cite{cacpor02} and other systems exploited it as well with robust,
asynchronous protocols.

For integrating a threshold cryptosystems with a BFT system, no particular
assumptions are needed about the structure of the atomic broadcast or even
the existence of a leader.  The distributed scheme can simply be inserted
into the code that executes operations and directly replaces the calls to
the cryptosystems.

To be more precise, suppose that $\op{execute}(s, o)$ invokes a call to one
of the three cryptosystem functions discussed before (that is, public-key
decryption, issuing a digital signature, or generating random bits).  The
process now invokes the threshold algorithm to generate a corresponding
share.  If the threshold cryptosystem is non-interactive, the process sends
this share to all others over the point-to-point links.  Then the process
waits for receiving $f+1$ shares and assembles them to the result of the
cryptographic operation.  With interactive threshold schemes, the processes
invoke the corresponding protocol as a subroutine.  Ideally the threshold
cryptosystem supports the same cryptographic signature structure or
ciphertext format as the standardized schemes; then the rest of the service
(i.e., code of the clients) can remain the same as with a centralized
service.  This holds for RSA signatures, for instance~\cite{cacsam04}.

\section{Conclusion}

This paper has introduced a distinction between three models for dealing
with \nondet operations in BFT replication: \emph{modular} where the
application is a black box; \emph{master-slave} that needs internal access
to the application; and \emph{cryptographically secure} handling of \nondet
randomness generation.  In the past, dedicated BFT replication systems have
often argued for using the master-slave model, but we have learned in the
context of blockchain applications that changes of the code and
understanding an application's logic can be difficult.  Hence, our novel
Protocol~\emph{Sieve} provides a modular solution that does not require any
manual intervention.  For a BFT-based blockchain platform, \emph{Sieve} can
simply be run without incurring large overhead as a defense against
non-determinism, which may be hidden in smart contracts.

\section*{Acknowledgments}

We thank our colleagues and the members of the IBM Blockchain development
team for interesting discussions and valuable comments, in particular Elli
Androulaki, Konstantinos Christidis, Angelo De Caro, Chet Murthy, Binh
Nguyen, and Michael Osborne.

This work was supported in part by the European Union's Horizon 2020
Framework Programme under grant agreement number~643964 (SUPERCLOUD) and in
part by the Swiss State Secretariat for Education, Research and Innovation
(SERI) under contract number~15.0091.


\begin{thebibliography}{10}

\bibitem{bmcnkf15}
J.~Bonneau, A.~Miller, J.~Clark, A.~Narayanan, J.~A. Kroll, and E.~W. Felten.
\newblock {SoK}: {Research} perspectives and challenges for {Bitcoin} and
  cryptocurrencies.
\newblock In {\em Proc.\ 36th IEEE Symposium on Security \& Privacy}, pages
  104--121, 2015.

\bibitem{bresch96}
T.~C. Bressoud and F.~B. Schneider.
\newblock Hypervisor-based fault-tolerance.
\newblock {\em ACM Transactions on Computer Systems}, 14(1):80--107, Feb. 1996.

\bibitem{bmst93}
N.~Budhiraja, K.~Marzullo, F.~B. Schneider, and S.~Toueg.
\newblock The primary-backup approach.
\newblock In {\em Distributed Systems (2nd Ed.)}. ACM Press \& Addison-Wesley,
  New York, 1993.

\bibitem{cachin01}
C.~Cachin.
\newblock Distributing trust on the {Internet}.
\newblock In {\em Proc.\ International Conference on Dependable Systems and
  Networks (DSN-DCCS)}, pages 183--192, 2001.

\bibitem{dccl16}
C.~Cachin, editor.
\newblock {\em Distributed Cryptocurrencies and Consensus Ledgers ({DCCL
  2016})}, Online proceedings of workshop co-located with PODC, 2016.
\newblock \url{https://www.zurich.ibm.com/dccl/}.

\bibitem{CachinGR11}
C.~Cachin, R.~Guerraoui, and L.~Rodrigues.
\newblock {\em Introduction to Reliable and Secure Distributed Programming
  ({Second Edition})}.
\newblock Springer, 2011.

\bibitem{ckps01}
C.~Cachin, K.~Kursawe, F.~Petzold, and V.~Shoup.
\newblock Secure and efficient asynchronous broadcast protocols (extended
  abstract).
\newblock In {\em Advances in Cryptology:\ CRYPTO 2001}, volume 2139 of {\em
  Lecture Notes in Computer Science}, pages 524--541. Springer, 2001.

\bibitem{cakush05}
C.~Cachin, K.~Kursawe, and V.~Shoup.
\newblock Random oracles in {Constantinople}: Practical asynchronous
  {Byzantine} agreement using cryptography.
\newblock {\em Journal of Cryptology}, 18(3):219--246, 2005.

\bibitem{cacpor02}
C.~Cachin and J.~A. Poritz.
\newblock Secure intrusion-tolerant replication on the {Internet}.
\newblock In {\em Proc.\ International Conference on Dependable Systems and
  Networks (DSN-DCCS)}, pages 167--176, June 2002.

\bibitem{cacsam04}
C.~Cachin and A.~Samar.
\newblock Secure distributed {DNS}.
\newblock In {\em Proc.\ International Conference on Dependable Systems and
  Networks (DSN-DCCS)}, pages 423--432, June 2004.

\bibitem{caslis02}
M.~Castro and B.~Liskov.
\newblock Practical {B}yzantine fault tolerance and proactive recovery.
\newblock {\em ACM Transactions on Computer Systems}, 20(4):398--461, Nov.
  2002.

\bibitem{caroli03}
M.~Castro, R.~Rodrigues, and B.~Liskov.
\newblock {BASE}: Using abstraction to improve fault tolerance.
\newblock {\em ACM Transactions on Computer Systems}, 21(3):236--269, 2003.

\bibitem{chatou96}
T.~D. Chandra and S.~Toueg.
\newblock Unreliable failure detectors for reliable distributed systems.
\newblock {\em Journal of the ACM}, 43(2):225--267, 1996.

\bibitem{cbpesc10}
B.~Charron-Bost, F.~Pedone, and A.~Schiper, editors.
\newblock {\em Replication: Theory and Practice}, volume 5959 of {\em Lecture
  Notes in Computer Science}.
\newblock Springer, 2010.

\bibitem{cklwad09}
A.~Clement, M.~Kapritsos, S.~Lee, Y.~Wang, L.~Alvisi, M.~Dahlin, and T.~Riche.
\newblock {UpRight} cluster services.
\newblock In {\em Proc.\ 22nd ACM Symposium on Operating Systems Principles
  (SOSP)}, pages 277--290, 2009.

\bibitem{cdegjk16}
K.~Croman, C.~Decker, I.~Eyal, A.~E. Gencer, A.~Juels, A.~Kosba, A.~Miller,
  P.~Saxena, E.~Shi, E.~G. Sirer, D.~Song, and R.~Wattenhofer.
\newblock On scaling decentralized blockchains.
\newblock In {\em Proc.\ 3rd Workshop on Bitcoin and Blockchain Research},
  2016.

\bibitem{desmed94}
Y.~Desmedt.
\newblock Threshold cryptography.
\newblock {\em European Transactions on Telecommunications}, 5(4):449--457,
  1994.

\bibitem{dggs99}
A.~Doudou, B.~Garbinato, R.~Guerraoui, and A.~Schiper.
\newblock Muteness failure detectors: Specification and implementation.
\newblock In {\em Proc.\ 3rd European Dependable Computing Conference
  (EDCC-3)}, volume 1667 of {\em Lecture Notes in Computer Science}, pages
  71--87. Springer, 1999.

\bibitem{haizha16}
S.~Duan and H.~Zhang.
\newblock Practical confidential state machine replication: How to process data
  privately in the cloud.
\newblock In {\em Proc.\ 35th Symposium on Reliable Distributed Systems
  (SRDS)}, 2016.

\bibitem{dwlyst88}
C.~Dwork, N.~Lynch, and L.~Stockmeyer.
\newblock Consensus in the presence of partial synchrony.
\newblock {\em Journal of the ACM}, 35(2):288--323, 1988.

\bibitem{filypa85}
M.~J. Fischer, N.~A. Lynch, and M.~S. Paterson.
\newblock Impossibility of distributed consensus with one faulty process.
\newblock {\em Journal of the ACM}, 32(2):374--382, Apr. 1985.

\bibitem{gjkr07}
R.~Gennaro, S.~Jarecki, H.~Krawczyk, and T.~Rabin.
\newblock Secure distributed key generation for discrete-log based
  cryptosystems.
\newblock {\em Journal of Cryptology}, 20:51--83, 2007.

\bibitem{ghyzzz14}
Z.~Guo, C.~Hong, M.~Yang, D.~Zhou, L.~Zhou, and L.~Zhuang.
\newblock Rex: Replication at the speed of multi-core.
\newblock In {\em Proc.\ 9th European Conference on Computer Systems
  (EuroSys)}, 2014.

\bibitem{hadtou93}
V.~Hadzilacos and S.~Toueg.
\newblock Fault-tolerant broadcasts and related problems.
\newblock In {\em Distributed Systems}. ACM Press \& Addison-Wesley, New York,
  1993.

\bibitem{jager15}
T.~Jager.
\newblock Verifiable random functions from weaker assumptions.
\newblock In {\em Proc.\ 12th Theory of Cryptography Conference (TCC 2015)},
  volume 9015 of {\em Lecture Notes in Computer Science}, pages 121--143.
  Springer, 2015.

\bibitem{kscsd10}
R.~Kapitza, M.~Schunter, C.~Cachin, K.~Stengel, and T.~Distler.
\newblock Storyboard: Optimistic deterministic multithreading.
\newblock In {\em Proc.\ 6th Workshop on Hot Topics in System Dependability},
  2010.

\bibitem{kwqcad12}
M.~Kapritsos, Y.~Wang, V.~Quema, A.~Clement, L.~Alvisi, and M.~Dahlin.
\newblock All about {Eve}: Execute-verify replication for multi-core servers.
\newblock In {\em Proc.\ 10th Symp.\ Operating Systems Design and
  Implementation (OSDI)}, 2012.

\bibitem{kotdah04}
R.~Kotla and M.~Dahlin.
\newblock High throughput {Byzantine} fault tolerance.
\newblock In {\em Proc.\ International Conference on Dependable Systems and
  Networks (DSN-DCCS)}, pages 575--584, June 2004.

\bibitem{liskov10}
B.~Liskov.
\newblock From viewstamped replication to {Byzantine} fault tolerance.
\newblock In {\em Replication: Theory and Practice}, volume 5959 of {\em
  Lecture Notes in Computer Science}, pages 121--149. Springer, 2010.

\bibitem{lysyan02}
A.~Lysyanskaya.
\newblock Unique signatures and verifiable random functions from the {DH}-{DDH}
  separation.
\newblock In {\em Advances in Cryptology:\ CRYPTO 2002}, volume 2442 of {\em
  Lecture Notes in Computer Science}, pages 597--612. Springer, 2002.

\bibitem{MenezesOV97}
A.~J. Menezes, P.~C. van Oorschot, and S.~A. Vanstone.
\newblock {\em Handbook of Applied Cryptography}.
\newblock CRC Press, Boca Raton, FL, 1997.

\bibitem{mirava99}
S.~Micali, M.~Rabin, and S.~Vadhan.
\newblock Verifiable random functions.
\newblock In {\em Proc.\ 40th IEEE Symposium on Foundations of Computer Science
  (FOCS)}, pages 120--130, 1999.

\bibitem{mxcss16}
A.~Miller, Y.~Xia, K.~Croman, E.~Shi, and D.~Song.
\newblock The honey badger of {BFT} protocols.
\newblock In {\em Proc.\ ACM Conference on Computer and Communications Security
  (CCS)}, 2016.

\bibitem{peshla80}
M.~Pease, R.~Shostak, and L.~Lamport.
\newblock Reaching agreement in the presence of faults.
\newblock {\em Journal of the ACM}, 27(2):228--234, Apr. 1980.

\bibitem{reibir94}
M.~K. Reiter and K.~P. Birman.
\newblock How to securely replicate services.
\newblock {\em ACM Transactions on Programming Languages and Systems},
  16(3):986--1009, May 1994.

\bibitem{schnei90}
F.~B. Schneider.
\newblock Implementing fault-tolerant services using the state machine
  approach: A tutorial.
\newblock {\em ACM Computing Surveys}, 22(4):299--319, Dec. 1990.

\bibitem{shoup00a}
V.~Shoup.
\newblock Practical threshold signatures.
\newblock In {\em Advances in Cryptology:\ EUROCRYPT 2000}, volume 1087 of {\em
  Lecture Notes in Computer Science}, pages 207--220. Springer, 2000.

\bibitem{shogen98}
V.~Shoup and R.~Gennaro.
\newblock Securing threshold cryptosystems against chosen ciphertext attack.
\newblock In {\em Advances in Cryptology:\ EUROCRYPT '98}, volume 1403 of {\em
  Lecture Notes in Computer Science}, pages 1--16. Springer, 1998.

\bibitem{swanso15}
T.~Swanson.
\newblock Consensus-as-a-service: A brief report on the emergence of
  permissioned, distributed ledger systems.
\newblock Report, Apr. 2015.
\newblock URL:
  \url{http://www.ofnumbers.com/wp-content/uploads/2015/04/Permissioned-distributed-ledgers.pdf}.

\bibitem{Vukolic16}
M.~Vukoli\'c.
\newblock The quest for scalable blockchain fabric: Proof-of-work vs. {BFT}
  replication.
\newblock In {\em Open Problems in Network Security, Proc.\ IFIP WG~11.4
  Workshop (iNetSec 2015)}, volume 9591 of {\em Lecture Notes in Computer
  Science}, pages 112--125. Springer, 2016.

\bibitem{walblu15}
M.~Walfish and A.~J. Blumberg.
\newblock Verifying computations without reexecuting them.
\newblock {\em Communications of the ACM}, 58(2), Feb. 2015.

\bibitem{zhscre02a}
L.~Zhou, F.~B. Schneider, and R.~van Renesse.
\newblock {COCA}: A secure distributed online certification authority.
\newblock {\em ACM Transactions on Computer Systems}, 20(4):329--368, 2002.

\end{thebibliography}

\end{document}